\documentclass[11pt]{llncs}

\usepackage{latexsym,amsfonts}
\usepackage{amsmath}
\usepackage{amssymb}
\usepackage{array}

\usepackage{graphicx,color}
\usepackage{latexsym,amsfonts}
\usepackage{amsmath}
\usepackage{amssymb}
\usepackage{bbm}

\usepackage{tikz}
\usepackage{url}
\usepackage{fullpage}

\usepackage{algorithm}
\usepackage{algorithmicx}
\usepackage[noend]{algpseudocode}
\usepackage{thmtools,thm-restate}

\usepackage[breaklinks=true]{hyperref} 
\hypersetup{
	colorlinks=true,
	citecolor=teal,
	linkcolor=blue,  
}

\newcommand{\wt}{\mathsf{wt}}
\newtheorem{new-claim}{Claim}

\begin{document}
\pagestyle{plain}
\title{Popular Matchings with One-Sided Bias\thanks{A preliminary version of this paper appeared in ICALP~2020~\cite{Kav20}.}}
\author{Telikepalli Kavitha\thanks{Work done while visiting Max-Planck-Institut f\"ur Informatik, Saarbr\"ucken, Germany.}}
\institute{Tata Institute of Fundamental Research, Mumbai, India \\ \email{kavitha@tifr.res.in}}
\maketitle

\begin{abstract}
  Let $G = (A \cup B,E)$ be a bipartite graph where the set $A$ consists of agents or {\em main players} and the set $B$
  consists of jobs or {\em secondary players}. Every vertex has a strict ranking of its neighbors. A matching $M$ is 
  {\em popular} if for any matching $N$, the number of  vertices that prefer $M$ to $N$ is at least the number that prefer 
  $N$ to $M$. Popular matchings always exist in $G$ since every stable matching is popular.

\smallskip

  A matching $M$ is {\em $A$-popular} if for any matching $N$, the number of  {\em agents} (i.e., vertices in $A$) that prefer
  $M$ to $N$ is at least the number of agents that prefer $N$ to $M$. Unlike popular matchings, $A$-popular matchings need not
  exist in a given instance $G$ and there is a simple linear time algorithm to decide if $G$ admits an $A$-popular matching and
  compute one, if so.

\smallskip

  We consider the problem of deciding if $G$ admits a matching that is both popular and $A$-popular and finding one, if so.
  We call such matchings {\em fully popular}. A fully popular matching is useful when $A$ is the more important side---so along with
  overall popularity, we would like to maintain ``popularity within~the~set~$A$''.
  A fully popular matching is not necessarily a min-size/max-size popular matching
  and all known polynomial-time algorithms for popular matching problems compute either min-size or max-size popular matchings.
  Here we show a linear time algorithm for the fully popular matching problem, thus our result shows a new tractable subclass of
  popular matchings.
\end{abstract}

\section{Introduction}
\label{sec:intro}

Let $G = (A \cup B, E)$ be a bipartite graph where vertices in $A$ are called {\em agents} and those in $B$ are called {\em jobs}.
Every vertex has a strict ranking of its neighbors. Such a graph, also called a {\em marriage} instance,
is a very well-studied model in two-sided matching markets. A matching $M$ in $G$ is {\em stable} if there is 
no {\em blocking edge} with respect to $M$, i.e., no edge $(a,b)$ such that $a$ and $b$ prefer each other to their
respective assignments in $M$. Gale and Shapley~\cite{GS62} in 1962 showed that
stable matchings always exist in $G$ and can be efficiently computed. 

Stable matching algorithms have applications in several real-world problems. For instance, stable matchings have been extensively
used to match students to schools and colleges \cite{AS03,BCCKP18} and one of the oldest applications here is to match medical
residents to hospitals~\cite{CRMS,NRMP}. It is known that all stable matchings in $G$ have the same size~\cite{GS85} and this may
only be half the size of a maximum matching in $G$. Consider the following instance on four vertices $a_0,a_1,b_0,b_1$.
The preferences of these four vertices are as follows:
\[a_0: b_1 \ \ \ \ \ \ \ \ \ \ \ \ a_1: b_1 \succ b_0 \ \ \ \ \ \ \ \ \ \ \ \ b_0: a_1 \ \ \ \ \ \ \ \ \ \ \ \ b_1: a_1 \succ a_0.\]
Here $a_1$ and $b_1$ are each other's top choices. There is no edge between $a_0$ and $b_0$.
Note that $M_{\max} = \{(a_0,b_1),(a_1,b_0)\}$ has size~2 while the only stable matching
$S = \{(a_1,b_1)\}$ has size~1.

Hence forbidding blocking edges constrains the size of the resulting matching. 
Rather than empower every edge with a ``veto power'' to block matchings (this is the notion of stability),
we would like to relax stability so that the power to block matchings gets diffused among all the vertices. That is, rather
than a single pair of vertices declaring that a given matching is infeasible, we now want all the vertices to 
participate in deciding if a given matching is feasible or not.
The motivation is to obtain a larger pool of feasible matchings
in order to obtain improved matchings with respect to size or any other desired objective.

\paragraph{\bf Popularity.}
The notion of {\em popularity} is a natural relaxation of stability that captures collective welfare. 
Consider an election between two matchings $M$ and $N$ where vertices are voters. Preferences of a vertex 
over its neighbors extend naturally to preferences over matchings---in the $M$ versus~$N$ election, each vertex 
votes for the matching in $\{M,N\}$ that it prefers, i.e., where it gets a better assignment.
Note that a vertex abstains from voting if it has the same assignment in both $M$~and~$N$;
also, being left unmatched is the worst choice for any vertex.
Let $\phi(M,N)$ (resp., $\phi(N,M)$) be the number of votes for $M$ (resp., $N$) in this election.

\begin{definition}
A matching $M$ is {\em popular} if $\phi(M,N) \ge \phi(N,M)$ for all matchings $N$ in $G$. 
\end{definition}

So a popular matching never loses a head-to-head election against any matching, in other words, it is a weak 
{\em Condorcet winner}~\cite{Con85,condorcet} in the voting instance where matchings are candidates and vertices are voters. 
The notion of popularity was introduced by G\"ardenfors~\cite{Gar75} who showed that every stable matching is popular. 
So popular matchings always exist in any marriage instance. In fact, every stable matching is a {\em min-size} popular matching~\cite{HK11}. Going to back to our earlier example on the four vertices $a_0,a_1,b_0,b_1$, the matching
$M_{\max} = \{(a_0,b_1),(a_1,b_0)\}$, though unstable, is popular.
There are efficient algorithms to compute a {\em max-size} popular matching~\cite{HK11,Kav14}.

Popular matchings are suitable in applications such as matching students to projects where students and project advisers 
have strict preferences. By relaxing stability to popularity, we can obtain better matchings in terms of size (as in the 
above example) or some other desired objective. We consider a natural and relevant objective here: observe that the two sides of 
$G = (A \cup B, E)$ are {\em asymmetric} in this application---students are the {\em doers} of the projects, i.e., they are 
the main or more active players, while project advisers are the secondary or more passive players. So along with overall
popularity, we would like to maintain ``popularity within~the~set~$A$''.

That is, we would like the popular matching that we compute 
to be popular even when we {\em only} count the votes of vertices in the set $A$, so there should be no matching that is preferred by more
vertices in $A$. Popularity within the set $A$ is the notion of popularity with {\em one-sided} preferences and we will refer
to this as {\em $A$-popularity} here. In the 
$M$ versus $N$ election, let  $\phi_A(M,N)$ (resp., $\phi_A(N,M)$) be the number of vertices in $A$ that vote for $M$ (resp., $N$).

\begin{definition}
A matching $M$ is {\em $A$-popular} if $\phi_A(M,N) \ge \phi_A(N,M)$ for all matchings $N$ in $G$.
\end{definition}

Matchings that are $A$-popular have been well-studied~\cite{AIKM07,KMN09,MS06,Mestre06} and $A$-popular matchings are relevant
in applications such as assigning training posts to applicants~\cite{AIKM07} and housing allocation schemes~\cite{MS06} where
vertices on only one side of the graph have preferences over their neighbors. An $A$-popular matching need not necessarily exist
in a given instance.
Consider the following instance on three agents $a_1,a_2,a_3$ where all the agents have identical preferences as shown below.
\[a_1: b_1 \succ b_2 \succ b_3 \ \ \ \ \ \ \ \ \ \ \ \ \ \ \ a_2: b_1 \succ b_2 \succ b_3 \ \ \ \ \ \ \ \ \ \ \ \ \ \ \ a_3: b_1 \succ b_2 \succ b_3.\]
It is easy to check that none of the matchings in the above instance is $A$-popular.
Let $M_0 = \{(a_1,b_1),(a_2,b_2),(a_3,b_3)\}$ and $M_1 = \{(a_1,b_3),(a_2,b_1),(a_3,b_2)\}$. We have
$\phi_A(M_1,M_0) = 2 > 1 = \phi_A(M_0,M_1)$ since $a_2$ and $a_3$ prefer $M_1$ to $M_0$ while $a_1$ prefers $M_0$ to $M_1$.
Similarly, for any matching~$M$ in this instance, there exists some other matching that is more ``$A$-popular'' than $M$.
Thus  the above instance has no $A$-popular matching.

We now seek matchings that are both popular {\em and} $A$-popular. So let us define the following subclass of popular 
matchings.

\begin{definition}
A popular matching $M$ in $G = (A \cup B, E)$ is {\em fully popular} if $M$ is also $A$-popular. So for any matching $N$ in $G$, we have:
$\phi(M,N) \ge \phi(N,M)$ and  $\phi_A(M,N) \ge \phi_A(N,M)$.
\end{definition}

There may be exponentially many popular matchings in $G = (A \cup B,E)$. So when
$A$ is the more important/active side, say it consists of those doing their projects/internships/jobs, it is natural 
to seek a popular matching that is $A$-popular as well, i.e., a {\em fully popular} matching.
We show the following result here.

\begin{restatable}{theorem}{MainTheorem}
\label{thm:algo}
There is a linear time algorithm to decide if a marriage instance $G = (A \cup B, E)$ with strict preferences admits a fully popular matching
or not. If so, our algorithm returns a max-size fully popular matching.
\end{restatable}

\subsection{Background and Related results}
The notion of popularity was proposed by G\"ardenfors~\cite{Gar75} in 1975. 
Algorithms in the domain of popular matchings were first studied in 2005 for {\em one-sided} preferences or 
the $A$-popular matching problem. 
Efficient algorithms were given in \cite{AIKM07} to decide if a given instance (with ties permitted in preferences) admits an 
$A$-popular matching or not; in particular, a linear time algorithm was given for the case with strict preferences. 

Algorithms for popular matchings in a marriage instance $G = (A \cup B,E)$ or {\em two-sided} preferences have been well-studied
in the last decade. The max-size popular matching algorithms in \cite{HK11,Kav14} compute special popular matchings called 
{\em dominant} matchings. A linear time algorithm for finding a popular matching with a given edge $e$ was given in \cite{CK16}
(such an edge is called a popular edge). It was shown in \cite{CK16} that if $e$ is a popular edge then there is either a
stable matching or a dominant matching with the edge $e$. 

Popular half-integral matchings in $G = (A \cup B, E)$ were characterized in \cite{Kav16} as stable matchings in a larger graph
related to $G$. The popular fractional matching polytope was analyzed in \cite{HK17} where the half-integrality of this polytope
was shown. Other than algorithms for min-size/max-size popular matchings and for the popular edge problem, no other 
polynomial-time algorithms were known for finding popular matchings with special properties.

To complete the picture, it was shown in \cite{FKPZ19} that it is NP-hard to decide if $G$ admits a popular matching that is
neither a min-size nor a max-size popular matching. A host of hardness results in~\cite{FKPZ19} painted a bleak picture for
efficient algorithms for popular matching problems (other than what is already known). For instance, it is NP-hard to find a
popular matching in $G$ with a given pair of edges. Thus finding a max-weight (resp., min-cost) popular matching is NP-hard 
when there are weights (resp., costs) on edges.

\subsection{Our Result and Techniques} 
It may be the case that no min-size or max-size popular matching in $G$ is $A$-popular, 
however $G$ admits a fully popular matching; Section~\ref{prelims} has such an example. 
As there are instances where it is NP-hard to decide if there exists a popular 
matching that is neither a min-size nor a max-size popular matching~\cite{FKPZ19}, 
a first guess may be that the fully popular matching problem is NP-hard.

Though an $A$-popular matching is constrained to use only some special edges in $G$ (see Theorem~\ref{thm:A-popular}),
this does not seem very helpful since it is NP-hard to solve the popular matching problem with forced edges~\cite{FKPZ19}.
Note that a rival matching 
is free to use any edge in $G$. It was not known if there was any tractable subclass of popular matchings 
other than the classes of {\em stable} matchings~\cite{GS62} and {\em dominant} matchings~\cite{CK16,HK11,Kav14}.

We show the set of fully popular matchings is a new tractable subclass of popular matchings. Unlike the classes of stable matchings and dominant matchings which are always non-empty, there need not exist a fully popular matching in $G$.
Our algorithm for finding a fully popular matching is based on the classical Gale-Shapley algorithm and works in a new graph $H$. This graph $H$
is essentially two copies of $G$ and is a variant of the graph seen in \cite{Kav16} to study popular {\em half-integral}
matchings. There is a natural map from the set of stable matchings in $H$ to the set of popular half-integral  matchings in $G$.
Our goal is to compute a stable matching with sufficient
symmetry in $H$ so that we can obtain a popular {\em integral} matching in $G$. 

We achieve this symmetry by using properties of both popular and $A$-popular matchings. These properties
allow us to identify certain edges that have to be excluded from our matching.
If there is no stable matching in $H$ without these edges then we use the lattice structure on stable matchings~\cite{GI89}
to show that $G$ has no fully popular matching. Else we obtain a matching $M$ in $G$ from this 
``partially symmetric'' stable matching in $H$. The most technical part of our analysis is to prove $M$'s popularity in $G$. 

\paragraph{Organization of the paper.} Section~\ref{prelims} discusses preliminaries on popular matchings and $A$-popular matchings.
Our algorithm is presented in Section~\ref{sec:algo} and its correctness is proved in Section~\ref{sec:correct}.

\section{Preliminaries}
\label{prelims}
Our input is a bipartite graph  $G = (A \cup B, E)$ where every vertex has a strict preference order on its neighbors.
Let us augment $G$ with self-loops, so each vertex is assumed to be at the bottom of its own preference list. Hence for any
vertex, being matched along a self-loop will be equivalent to what was originally the state of being left unmatched.
Thus popularity (resp., $A$-popularity) in the augmented instance is equivalent to 
popularity (resp., $A$-popularity) in the original instance.

We will first present the characterization of $A$-popular matchings---note that preferences of vertices 
in $B$ play no role here. For each $a \in A$, define the vertex $f(a)$ to be $a$'s top 
choice neighbor and let $s(a)$ be $a$'s most preferred neighbor that is nobody's top choice neighbor. 
We assume every $a \in A$ has at least one neighbor other than itself, so $f(a) \in B$, however it may be the case that 
$s(a) = a$. The following characterization of $A$-popular matchings was given in \cite{AIKM07}.
Let $E' = E \cup \{(u,u): u \in A \cup B\}$. 

\begin{theorem}[\cite{AIKM07}]
\label{thm:A-popular}
A matching $M$ in $G = (A \cup B, E')$ is $A$-popular if and only if:
\begin{enumerate}
\item $M \subseteq \{(a,f(a)), (a,s(a)): a \in A\}$. 
\item $M$ matches all in $A$ and all in $\{f(a): a \in A\}$.
\end{enumerate}
\end{theorem}

Thus any $A$-popular matching $M$ has to match every $a \in A$ to either $f(a)$ or $s(a)$. Furthermore, any job $b \in B$ that is some agent's
top choice neighbor has to be matched in $M$ to an agent $a \in A$ such that $b = f(a)$, i.e., $a$ regards $b$ as its top choice neighbor.

\paragraph{\bf Popular matchings.} We will use an LP-based characterization of popular matchings~\cite{Kav16,KMN09} in 
a marriage instance $G$. Recall that we augmented the edge set
$E$ with self-loops. It will be convenient to view any matching in the original instance as a perfect matching in the augmented
instance $G = (A \cup B, E')$ by using self-loops to match all the vertices originally left unmatched.

Let $M$ be any perfect matching in~$G= (A \cup B, E')$. For any vertex $u$, let $M(u)$ be $u$'s partner in $M$. 
For any pair of adjacent vertices $u$ and $v$, let $u$'s vote for $v$ versus $M(u)$ be $1$ if $u$ prefers $v$ to $M(u)$, 
it is $-1$ if $u$ prefers $M(u)$ to $v$, else it is 0 (in this case $M(u) = v$). 
In order to check if $M$ is popular or not in $G$, the following 
edge weight function $\wt_M$ will be useful. Note that $\wt_M(a,b)$ is the sum of votes of $a$ and $b$ for each other
versus their respective assignments in $M$. 

For any $(a,b) \in E$:

\begin{equation*} 
\wt_M(a,b) = \begin{cases} 2   & \text{if\ $(a,b)$\ is\ a\ blocking\ edge\ to\ $M$;}\\
	                     -2 &  \text{if\ both\ $a$\ and\ $b$\ prefer\ their\ partners\ in\ $M$\ to\ each\ other;}\\			
                              0 & \text{otherwise.}
\end{cases}
\end{equation*}

Thus $\wt_M(a,b) = 0$ for every $(a,b) \in M$.
We need to define $\wt_M$ for self-loops as well. For any $u \in A \cup B$:
\begin{equation*} 
\wt_M(u,u) = \begin{cases} 0   & \text{if\ $(u,u) \in M$;}\\
                              -1 & \text{otherwise.}
\end{cases}
\end{equation*}

For any perfect matching $N$ in~$G$, observe that $\wt_M(N) = \sum_{e\in N}\wt_M(e) = \phi(N,M) - \phi(M,N)$.
Thus $M$ is popular if and only if $\wt_M(N) \le 0$ for every perfect matching $N$ in $G$.

Consider the max-weight perfect matching LP in $G$ with the edge weight function $\wt_M$. 
This linear program is (LP1) given below and (LP2) is the dual of (LP1).
The variables $x_e$ for $e \in E'$ are primal variables and the variables $y_u$ for $u \in A \cup B$ are dual variables.
Here $\delta'(u) = \delta(u) \cup \{(u,u)\}$.

\begin{table}[ht]
\begin{minipage}[b]{0.45\linewidth}\centering
\begin{eqnarray*}
 \max \sum_{e \in E'} \wt_M(e)\cdot x_e  &\mbox{\hspace*{0.1in}}\text{(LP1)}\\
      \text{s.t.}\qquad\sum_{e \in \delta'(u)}x_e = 1  &\mbox{\hspace*{0.1in}}\forall\, u \in A \cup B\\
                        x_e  \ge 0   &\mbox{\hspace*{-0.05in}}\forall\, e \in E'. 
\end{eqnarray*}
\end{minipage}
\hspace{0.6cm}
\begin{minipage}[b]{0.45\linewidth}\centering
  \begin{eqnarray*}
\min \sum_{u \in A \cup B}y_u  &\mbox{\hspace*{0.1in}}\text{(LP2)}\\
       \text{s.t.}\qquad y_{a} + y_{b} \ge \wt_{M}(a,b)  &\mbox{\hspace*{0.1in}}\forall\, (a,b)\in E\\
                          y_u \ge \wt_M(u,u) &\mbox{\hspace*{0.19in}}\forall\, u \in A \cup B.\\
\end{eqnarray*}
\end{minipage}
\end{table}

$M$ is popular if and only if the optimal value of (LP1) is at most 0. In fact, the optimal value is exactly~0
since $M$ is a perfect matching in $G$ and $\wt_M(M) = 0$.
Thus $M$ is popular if and only if the optimal value of (LP2) is 0 (by LP-duality). 

\begin{theorem}[\cite{Kav16,KMN09}]
\label{thm:witness}
A matching $M$ in $G = (A \cup B, E')$ is popular if and only if there exists $\vec{\alpha} \in \{0, \pm 1\}^n$ (where $|A\cup B| = n$) 
such that $\sum_{u \in A \cup B}\alpha_u = 0$ along with 
\[\alpha_{a} + \alpha_{b} \ge \wt_M(a,b) \ \ \ \forall (a,b) \in E \ \ \ \ \ \ \ \ \ \text{and}\ \ \ \ \ \ \ \ \ \alpha_u \ge \wt_M(u,u) \ \ \ \forall u \in A \cup B.\]
\end{theorem}
\begin{proof}
The constraint matrix of (LP2) is totally unimodular. This is because $E$ is the edge set of a bipartite graph and adding self-loops 
preserves the total unimodularity of the constraint matrix. So (LP2) admits an optimal solution 
that is integral. Let $\vec{\alpha}$ be an integral optimal solution of (LP2). Hence 
$\vec{\alpha} \in \mathbb{Z}^n$. We need to show that $\vec{\alpha} \in \{0, \pm 1\}^n$.

We have $\alpha_u \ge \wt_M(u,u) \ge -1$ for all $u \in A \cup B$. 
Since $M$ is an optimal solution to (LP1), complementary slackness implies that
$\alpha_u + \alpha_v = \wt_M(u,v) = 0$ for each edge $(u,v) \in M$. 
Thus $\alpha_u = -\alpha_v \le 1$ for every vertex $u$ matched to a non-trivial neighbor $v$ in $M$. 

Regarding any vertex $u$ such that $(u,u) \in M$, we again have by complementary slackness $\alpha_u = \wt_M(u,u) = 0$.
Hence $\vec{\alpha} \in \{0, \pm 1\}^n$. \qed
\end{proof}

For any popular matching $M$, a vector $\vec{\alpha}$ as given in Theorem~\ref{thm:witness} will 
be called a {\em witness} of $M$'s popularity. A popular matching may have several witnesses. 
A stable matching $S$ in $G$ has $\vec{\alpha} = 0^n$ as a witness since $\wt_S(e) \le 0$ for all $e \in E'$.

\paragraph{\bf An interesting example.}
Recall that our problem is to compute a {\em fully popular} matching, i.e., a popular matching that is also $A$-popular. 
It is easy to construct instances that admit $A$-popular matchings but admit no fully popular matching.
It could also be the case that no min-size or max-size popular matching in $G = (A \cup B, E)$ is $A$-popular, however $G$ has a fully 
popular matching. Consider the instance $G$ given in Fig.~\ref{fig:neither}. Vertex preferences are indicated on edges: 1 denotes top choice,
2 denotes second choice, and so on.

\begin{figure}[ht]
\centerline{\resizebox{0.75\textwidth}{!}{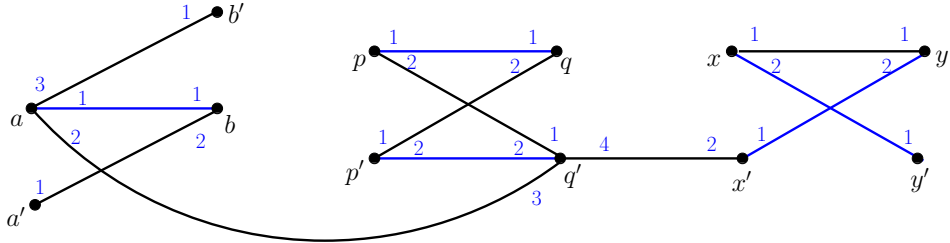}}
\caption{An instance on $A =\{a,a',p,p',x,x'\}$ and $B = \{b,b',q,q',y,y'\}$ where no min-size/max-size popular matching is $A$-popular. There is a fully popular matching (on blue edges) here.}
\label{fig:neither}
\end{figure}

We list the vertices $f(u)$ and $s(u)$ for each $u \in A$ in this instance.
Observe that the vertex $b'$ is {\em not} $s(a)$ since $a$ prefers $q'$ to $b'$ and $q' \ne f(u)$ for any $u \in A$.
\begin{itemize}
\item We have $f(a) = f(a') = b$, $f(p) = f(p') = q$, and $f(x) = f(x') = y$.
\item We have $s(a) = s(p) = s(p') = s(x') = q'$, $s(x) = y'$, and $s(a') = a'$.
\end{itemize}
Since $s(u) \ne u$ for $u \in \{a, p, p', x, x'\}$, any $A$-popular matching $M$ has to match these 
5 vertices to neighbors in $B$~(by Theorem~\ref{thm:A-popular}). So $q'$, which is $s(a)$, has to be matched to one of $p, p'$. Thus 
$M(a) = f(a) = b$ which implies that
$M(a') = s(a') = a'$. That is, after pruning self-loops from $M$, the vertex $a'$ has to be left unmatched in $M$. So $M$ has size~5.

The matching $S = \{(a,b),(p,q),(p',q'),(x,y)\}$ is stable. Thus any min-size popular matching in $G$ has size~4.
The perfect matching $M_{\max} = \{(a,b'),(a',b), (p,q),(p',q'),(x,y'),(x',y)\}$
is popular, so any max-size popular matching in $G$ has size 6. Thus no min-size or max-size popular matching in $G$ 
can be $A$-popular. Interestingly, this instance admits a fully popular matching; it is easy to check that the matching 
$M = \{(a,b), (p,q),(p',q'),(x,y'),(x',y)\}$ is both popular and $A$-popular.

\section{Fully Popular Matchings}
\label{sec:algo}

We are given a marriage instance $G = (A \cup B, E)$. Recall that we augmented $G$ with self-loops. 
So henceforth $G = (A \cup B, E')$ where
$E' = E \cup \{(u,u): u \in A \cup B\}$. Our algorithm will work in a bipartite graph $H$ which is essentially two copies of the 
graph $G$ as shown in Fig.~\ref{fig:H}. The vertex set of $H$ is $A_L \cup B_L$ on the left and $B_R \cup A_R$ on the right.
Here $A_L = \{a_{\ell}: a \in A\}$ and $A_R = \{a_r: a \in A\}$. Similarly, $B_L = \{b_{\ell}: b \in B\}$ and $B_R = \{b_r: b \in B\}$.

\begin{figure}[ht]
\centerline{\resizebox{0.42\textwidth}{!}{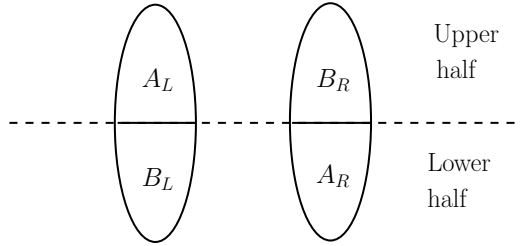}}
\caption{The bipartite graph $H$ consists of two copies of the graph $G = (A\cup B,E)$.}
\label{fig:H}
\end{figure}

The {\em upper half} of $H$ consists of the set $A_L$ of agents on the left and the set $B_R$ of jobs on the right 
while the {\em lower half} of $H$ consists of the set $B_L$ of jobs on the left and the set $A_R$ of agents 
on the right. Thus every vertex $u \in A \cup B$ has two copies in $H$: one as $u_{\ell}$ on the left of $H$ and
another as $u_r$ on the right of $H$.

For every edge in $E$, there will be four edges in $H$: a pair of parallel edges in the upper half and a pair of parallel edges in the lower half. In order to distinguish two parallel edges with the same endpoints, we use superscripts $+$ and $-$ on the endpoints.
For any $(a,b) \in E$:
\begin{itemize}
\item in the upper half, there are two parallel edges $(a_{\ell}^+,b_r^-)$ and $(a_{\ell}^-,b_r^+)$ between $a_{\ell}$ and $b_r$;
\item in the lower half, there are two parallel edges $(b_{\ell}^+,a_r^-)$ and $(b_{\ell}^-,a_r^+)$ between $b_{\ell}$ and $a_r$. 
\end{itemize}

\smallskip

\noindent{\em Remark.}
A pair of parallel edges $(u^+,v^-)$ and $(u^-,v^+)$ can be visualized as a bidirected pair $u \rightarrow v$ and $u \leftarrow v$.
We prefer to use superscripts instead of directions since these $+/-$ superscripts will be related to {\em witnesses} 
of popular matchings (see Lemma~\ref{lem:main}).

\medskip

Corresponding to every self-loop $(u,u)$, there is a {\em single} edge $(u_{\ell}^-,u_r^+)$ in $H$.
For convenience, we have used $+/-$ superscripts on the endpoints of this edge also. 
These edges $(u_{\ell}^-,u_r^+)$ for all $u \in A \cup B$ are the only edges in $H$ that go across the two halves of $H$.

Vertices in $H$ have preferences on their incident {\em edges} rather than on their neighbors. However it would be more
convenient to say $u$ prefers $v^-$ to $w^+$ rather than say  $u$ prefers $(u^+,v^-)$ to $(u^-,w^+)$.
In fact, $H$ is equivalent to a conventional graph $H^*$ (with preferences on neighbors) that was used 
to study popular {\em half-integral} matchings in \cite{Kav16}: there were 4 vertices in $H^*$ for each $u \in A \cup B$.
The graph $H$ is a sparser version of $H^*$ with only 2 vertices $u_{\ell}$ and $u_r$ for each $u \in A \cup B$ and a pair 
of parallel edges between every pair of adjacent vertices. We now describe the preferences of vertices in $H$.

\smallskip

\noindent{\bf Preferences in $H$.} 
Every vertex prefers superscript~$-$ neighbors to 
superscript~$+$ neighbors; among superscript~$-$ neighbors (similarly, superscript~$+$ neighbors), it will be its original preference order.
Consider any vertex $u \in A \cup B$. Suppose $u$'s preference order in the original instance $G = (A \cup B, E)$ is 
$v \succ v' \succ \cdots \succ v''$, i.e., $u$'s top choice is $v$, second choice is $v'$, and so on.
In $H$, the preference order of $u_{\ell}$ is as follows:
\[ \underbrace{v^-_r \succ v'^-_{r} \succ \cdots \succ v''^-_{r}}_{\text{superscript\ $-$\ neighbors}} \succ  \underbrace{v^+_{r} \succ v'^+_{r} \succ \cdots \succ v''^+_{r} \succ u^+_r}_{\text{superscript\ $+$\ neighbors}},\]
where $v_{r}, v'_r,\ldots$ correspond to the copies of $v, v', \ldots$ on the right side of $H$.

Observe that the vertex $u^+_r$ is the last choice of $u_{\ell}$. In $H$, the preference order of $u_r$ is as follows:
\[ \underbrace{v^-_{\ell} \succ v'^-_{\ell} \succ \cdots \succ v''^-_{\ell} \succ u^-_{\ell}}_{\text{superscript\ $-$\ neighbors}} \succ \underbrace{v^+_{\ell} \succ v'^+_{\ell} \succ \cdots \succ v''^+_{\ell}}_{\text{superscript\ $+$\ neighbors}},\]
where $v_{\ell}, v'_{\ell},\ldots$ correspond to the copies of $v, v', \ldots$ on the left side of $H$.
This is analogous to $u_{\ell}$'s preference order---the main difference is in the position of its {\em twin}---note that 
the vertex $u_r$ prefers $u^-_{\ell}$ to all its superscript $+$ neighbors.

\paragraph{\bf Blocking edges.}
For any matching $M$ in $H$, we say an edge $(u_{\ell}^+,v_r^-)$, where $u_{\ell} \in A_{\ell} \cup B_{\ell}$ 
and $v_r \in A_r \cup B_r$, {\em blocks} $M$ if the following two conditions hold:
\begin{enumerate}
    \item $u_{\ell}$ prefers $v_r^-$ to its assignment in $M$ and
    \item $v_r$ prefers $u_{\ell}^+$ to its assignment in $M$. 
\end{enumerate}

Similarly, we say $(u_{\ell}^-,v_r^+)$ {\em blocks} $M$ if $u_{\ell}$ prefers $v_r^+$ to its assignment in $M$ and $v_r$ prefers $u_{\ell}^-$ to 
its assignment in $M$.

\begin{definition}
  A matching $M$ in $H$ is {\em stable} if no edge in $H$ blocks $M$.
\end{definition}

For any perfect matching $S$ in $G = (A \cup B, E')$, there is a corresponding matching $S'$ in $H$ where
$S' = \{(a^-_{\ell},b^+_r),(b^-_{\ell},a^+_r): (a,b) \in S\cap E\} \cup \{(u^-_{\ell},u^+_r): (u,u) \in S\}$. 
So $S'$ is a perfect matching in $H$. The following claim will be useful to us.

\begin{new-claim}
\label{clm6}
If $S$ is stable in $G$ then $S'$ is stable in $H$. 
\end{new-claim}
\begin{proof}
We need to show that no edge in $H$ blocks $S'$. Consider any edge $(a^+_{\ell},b^-_r)$ in $H$ where $a \in A$ and $b \in B$. 
By the definition of the matching $S'$, some edge $(\ast,b^+_r) \in S'$.
Since every vertex prefers superscript $-$ neighbors to 
superscript~$+$ neighbors, the vertex $b_r$ prefers its partner in $S'$ to $a^+_{\ell}$. 
Thus the edge $(a^+_{\ell},b^-_r)$ does not block $S'$.

So consider any edge $(a^-_{\ell},b^+_r)$ in $H$. 
If $(a,b) \in S$ then $(a^-_{\ell},b^+_r) \in S'$ and so it does not block $S'$.
If $(a,b) \notin S$ then it follows from the stability of $S$ in $G$ that:
\begin{enumerate}
    \item either $a$ is matched to a neighbor $d$ preferred to $b$
    \item or $b$ is matched to a neighbor $c$ preferred to~$a$.
\end{enumerate}
In the first case, $a_{\ell}$ is matched in $S'$ to a neighbor $d^+_r$ preferred to $b^+_r$; 
in the second case, $b_r$ is matched in $S'$ to a neighbor $c^-_{\ell}$ preferred to $a^-_{\ell}$. Thus $(a^-_{\ell},b^+_r)$ does not block $S'$. 

It can analogously be shown that neither $(b^-_{\ell},a^+_r)$ nor $(b^+_{\ell},a^-_r)$ blocks $S'$. 
Also $(u^-_{\ell},u^+_r)$ for any $u \in A \cup B$ does not block $S'$ since $u^+_r$ is
$u_{\ell}$'s least preferred neighbor in $H$. Hence $S'$ is a stable matching in $H$. \qed
\end{proof}

Thus the graph $H$ admits a perfect stable matching. Since all stable matchings in $H$ have the same size~\cite{GS85}, 
every stable matching in $H$ has to be perfect.
We seek to compute a ``special'' stable matching in $H$: one that has no edge that is {\em forbidden}.
The edges that will be marked forbidden are those that no
fully popular matching can use. The definition of {\em valid} edges given below is as given in 
Theorem~\ref{thm:A-popular}: any $A$-popular matching in $G$ has to contain only these edges/self-loops.

\begin{definition}
\label{def:valid}
Edges/self-loops in $\{(a,f(a)), (a,s(a)): a \in A\}$ are {\em valid}. So are self-loops in 
$\{(b,b): b \ne f(a)\ \mathrm{for\ any}\ a \in A\}$. All other edges and self-loops are {\em invalid}. 
\end{definition}

Thus every $a \in A$ has exactly two valid edges incident to it: one of these may be the self-loop $(a,a)$.
Note that a job $b \in B$ may have several valid edges incident to it.

An edge $e$ in $G$ is called {\em popular} if there exists a popular matching $M$ in $G$ that contains~$e$. 
Similarly, a vertex is called {\em stable} if it is matched in some (equivalently, every~\cite{GS85}) stable matching.
It is known that every popular matching in $G$ has to match all stable vertices 
to {\em genuine} neighbors~\cite{HK11}. So the self-loop $(u,u)$ is popular if and only if $u$ is unstable.

\begin{definition}
\label{def:legal-edge}
Call an edge $e$ in $E' = E \cup \{(u,u): u \in A \cup B\}$ {\em legal} if $e$ is valid and popular.
\end{definition}

\paragraph{\bf Forbidden edges.}
A fully popular matching, by definition, has to contain only legal edges.
So if $(a,b) \in E$ is not legal then $(a_{\ell}^+,b_r^-)$, $(a_{\ell}^-,b_r^+)$, $(b_{\ell}^+,a_r^-)$, and 
$(b_{\ell}^-,a_r^+)$ are {\em forbidden} edges in the stable matching that we seek to compute in $H$.
Similarly, for any $u \in A \cup B$,
if $(u,u)$ is not legal then $(u_{\ell}^-,u_r^+)$ is a forbidden edge in our matching.

\begin{definition}
\label{def:legal}
A matching $M$ in $H$ is {\em legal} if $M$ has no forbidden edge.
\end{definition}

\paragraph{\bf Symmetric matchings.}
Call a matching $M$ in $H$ {\em symmetric} if for each edge $(a,b)$ in $E$, either both $(a_{\ell},b_r)$ and $(b_{\ell},a_r)$ are 
in $M$ or neither is in $M$. For convenience, we are not mentioning the $+/-$ superscripts on $a_{\ell},a_r,b_{\ell},b_r$. Loosely speaking, a symmetric matching $M$ has the same edges in the upper and lower 
halves of $H$. A symmetric matching $M$ in $H$ will be called a {\em realization} of 
$\tilde{M} = \{(a,b): (a_{\ell},b_r)\ \text{and}\ (b_{\ell},a_r)\ \text{are\ in}\ M\}$. 
Note that $\tilde{M}$ is a matching in $G$.

\begin{lemma}
\label{lem:main}
Let $N$ be any fully popular matching in $G$. Then there exists a legal stable matching in~$H$
that is a realization of $N$.
\end{lemma}
\begin{proof}
Let $N$ be a fully popular matching  in $G$ and let $\vec{\alpha}\in\{0, \pm 1\}^n$ be a witness of $N$'s popularity (see Theorem~\ref{thm:witness}). 
For $i \in \{0, \pm 1\}$, let $A_i$ be the set of vertices $a\in A$ with $\alpha_a = i$ and let $B_i$ be the set of vertices $b\in B$ 
with $\alpha_b = i$. Thus we have $A = A_0 \cup A_1 \cup A_{-1}$ and $B = B_0 \cup B_1 \cup B_{-1}$.

We have $\alpha_a + \alpha_b = \wt_N(a,b) =0$ for each edge $(a,b) \in N$: this is by complementary slackness on the linear program that is
analogous to (LP2) (so $\wt_N$ replaces $\wt_M$ in this LP). 
Since $\alpha_a = -\alpha_b$ for every $(a,b) \in N$, we have $N\cap E \subseteq (A_0 \times B_0) \cup (A_{-1} \times B_1) \cup (A_1 \times B_{-1})$
(see Fig.~\ref{fig:first}).

\begin{figure}[ht]
\centerline{\resizebox{0.55\textwidth}{!}{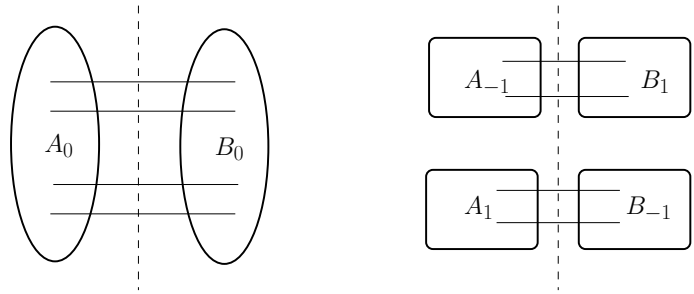}}
\caption{The partition $A_0 \cup A_1 \cup A_{-1}$ of $A$ and  $B_0 \cup B_{-1} \cup B_1$ of $B$.} 
\label{fig:first}
\end{figure}

We need to show a realization of $N$ in $H$ that is stable.
We will use $N$'s witness $\vec{\alpha}$ in $G$ to define the following symmetric matching $N^*_{\alpha}$ in $H$. Note that
this is similar to how popular half-integral matchings were realized as stable matchings in a larger graph as shown in \cite{Kav16}.
\begin{itemize}
\item For all $(a,b) \in N \cap (A_{-1} \times B_1)$ do: add edges $(a_{\ell}^-,b_r^+)$ and $(b_{\ell}^+,a_r^-)$ to $N^*_{\alpha}$.
\item For all $(a,b) \in N \cap (A_1 \times B_{-1})$ do: add edges $(a_{\ell}^+,b_r^-)$ and $(b_{\ell}^-,a_r^+)$ to $N^*_{\alpha}$.
\item For all $(a,b) \in N \cap (A_0 \times B_0)$ do: add edges $(a_{\ell}^-,b_r^+)$ and $(b_{\ell}^-,a_r^+)$ to $N^*_{\alpha}$. 
\end{itemize}

For each vertex $u$ such that $(u,u) \in N$, add $(u_{\ell}^-,u_r^+)$ to $N^*_{\alpha}$. 
Thus $N^*_{\alpha}$ is a perfect matching in~$H$.
Observe that for any vertex $u$, the sum of superscripts (where $\pm$ are interpreted as $\pm 1$) of the vertices
$u_{\ell}$ and~$u_r$ on the edges in $N^*_{\alpha}$ incident to them  is exactly $2\alpha_u$.

\begin{new-claim}
\label{clm3}
The matching $N^*_{\alpha}$ is stable in $H$.
\end{new-claim}

The proof of Claim~\ref{clm3} is based on the constraints that $\vec{\alpha}$ has to satisfy and is given below.
Moreover, the fact that $N$ is a fully popular matching in $G$ implies that $N^*_{\alpha}$ is a {\em legal} matching in $H$.
This is because every edge used in $N$ is valid (resp., popular) as $N$ is $A$-popular (resp., popular) in $G$. 
So $N^*_{\alpha}$ has no forbidden edge. Thus $N^*_{\alpha}$ is a legal stable matching in $H$.  \qed
\end{proof}

\paragraph{\bf Proof of Claim~\ref{clm3}.}
We need to show that $N^*_{\alpha}$ is a stable matching in $H$. Let us consider any edge $(a^-_{\ell},b^+_r)$ in $H$ and show that this 
edge does not block $N^*_{\alpha}$.
We can assume $a \ne b$ since the edge  $(a^-_{\ell},a^+_r)$ cannot block $N^*_{\alpha}$ as $a^+_r$ is
$a_{\ell}$'s least preferred neighbor in $H$. Thus we can assume $(a,b) \in E$.

Recall that $N$ is a perfect matching (due to augmenting $G$ with self-loops).
Let  $c = N(a)$ and $d = N(b)$. The matching $N^*_{\alpha}$ in $H$ was constructed using a witness 
$\vec{\alpha}$ of $N$'s popularity in $G$. We have the following cases depending on the values of $\alpha_a$ and $\alpha_b$. 

\begin{enumerate}
\item Suppose $\alpha_a = 1$. This means $(a_{\ell}^+,c^-_r) \in N^*_{\alpha}$ for some neighbor $c$ of $a$.
Since every vertex prefers superscript~$-$ neighbors to superscript~$+$ neighbors, the vertex $a_{\ell}$ prefers $c^-_r$ to 
$b^+_r$. Thus the edge $(a^-_{\ell},b^+_r)$ does not block $N^*_{\alpha}$.
Note that this argument is independent of the value of $\alpha_b$, so it holds for all $\alpha_b \in \{\pm 1, 0\}$.

\smallskip

\item Suppose $\alpha_a = 0$ and $\alpha_b = 1$. Then $(a_{\ell}^-,c^+_r)$ and $(d^-_{\ell},b^+_r)$ are in $N^*_{\alpha}$  for some neighbors $c$ and $d$ of $a$ 
and $b$, respectively. The edge covering constraint for
the edge $(a,b)$ tells us that $\alpha_a + \alpha_b = 1 \ge \wt_N(a,b)$. Since $\wt_N(a,b) \in \{0, \pm 2\}$, this means $\wt_N(a,b) \le 0$.
So either (i)~$a$ prefers $c$ to $b$ or (ii)~$b$ prefers $d$ to $a$. In the former case, $a_{\ell}$ prefers $c^+_r$ to $b^+_r$ and in the
latter case, $b_r$ prefers $d_{\ell}^-$ to $a_{\ell}^-$. Thus the edge $(a^-_{\ell},b^+_r)$ does not block $N^*_{\alpha}$.

\smallskip

\item Suppose $\alpha_a = \alpha_b = 0$. Either $(a,b) \in N$ which implies $(a^-_{\ell},b^+_r) \in N^*$ and so it does not block $N^*_{\alpha}$
or as analyzed in case~2 above, $(a_{\ell}^-,c^+_r)$ and $(d^-_{\ell},b^+_r)$ are in $N^*_{\alpha}$. Since $\wt_N(a,b) \le 0$, either (i)~$a$ prefers 
$c$ to $b$ or (ii)~$b$ prefers $d$ to $a$. So either $a_{\ell}$ prefers $c^+_r$ to $b^+_r$ or
$b_r$ prefers $d_{\ell}^-$ to $a_{\ell}^-$. Thus the edge $(a^-_{\ell},b^+_r)$ does not block $N^*_{\alpha}$.

\smallskip

\item Suppose $\alpha_a \in \{0, -1\}$ and $\alpha_b = -1$. Then $(a_{\ell}^-,c^+_r)$ and $(d^+_{\ell},b^-_r)$ are in $N^*_{\alpha}$ for some neighbors $c$ and $d$ 
of $a$ and $b$, respectively. The edge covering constraint for
$(a,b)$ tells us that $\alpha_a + \alpha_b = -1 \ge \wt_N(a,b)$. Since $\wt_N(a,b) \in \{0, \pm 2\}$, this means $\wt_N(a,b) \le -2$.
So {\em both} $a$ and $b$ prefer their partners in $N$ to each other. Hence $a_{\ell}$ prefers $c^+_r$ to $b^+_r$ and so 
$(a^-_{\ell},b^+_r)$ does not block $N^*_{\alpha}$.

\smallskip

\item Suppose $\alpha_a = -1$ and $\alpha_b \in \{0,1\}$. Then $(a_{\ell}^-,c^+_r)$ and $(d^-_{\ell},b^+_r)$ are in $N^*_{\alpha}$  for some neighbors 
$c$ and $d$ of $a$ and $b$, respectively.  The edge covering constraint for
the edge $(a,b)$ is $\alpha_a + \alpha_b = 0 \ge \wt_N(a,b)$. If $(a,b) \in N$ then $(a^-_{\ell},b^+_r) \in N^*$ and it does not block $N^*_{\alpha}$.
Otherwise either (i)~$a$ prefers $c$ to $b$ or (ii)~$b$ prefers $d$ to $a$. In the 
former case, $a_{\ell}$ prefers $c^+_r$ to $b^+_r$ and in the latter case, $b_r$ prefers $d_{\ell}^-$ to $a_{\ell}^-$. So $(a^-_{\ell},b^+_r)$ does not block $N^*_{\alpha}$.
\end{enumerate}

Thus the edge $(a^-_{\ell},b^+_r)$ does not block $N^*_{\alpha}$.
Analogous arguments show that {\em no} edge in $H$ blocks the matching $N^*_{\alpha}$.
Hence we can conclude that $N^*_{\alpha}$ is a stable matching in $H$. \qed

\paragraph{\bf Stable matchings with forbidden edges.}
A stable matching that avoids all forbidden edges (if such a matching exists in $H$) can be computed in linear time by running 
a variant of the Gale-Shapley algorithm in $H$ where any proposal made along a forbidden edge is rejected by the vertex receiving
this proposal. Once a proposal received
along a forbidden edge is rejected by a vertex, all further proposals received along {\em worse} edges also have to rejected by 
this vertex. If some vertex is left unmatched at the end of this algorithm, then there is no stable matching in $H$ that avoids
all forbidden edges; else we have a desired stable matching in $H$. We refer to \cite{GI89} for details on this variant of the
Gale-Shapley algorithm.

Thus it can be efficiently checked if $H$ admits a legal stable matching or not.
If such a matching does not exist in $H$ then there is no fully popular matching in $G$ (by Lemma~\ref{lem:main}). 
So we will assume henceforth that there exists a legal stable matching in $H$. However the fact that such a stable matching exists
in $H$ does not imply that $G$ admits a fully popular matching. This is because $H$ is made up of two copies of $G$, thus 
any matching $M^*$ in $H$ can only be mapped to a {\em half-integral} matching in $G$. 

In order to claim the resulting matching in $G$ is integral, we need $M^*$ to be symmetric, i.e., we need $M^*$ to
have the same edges in both halves of $H$. We will not construct such a symmetric stable matching in $H$. The matching we compute 
will have a certain amount of symmetry and this will be enough to obtain a fully popular matching in $G$. If $H$ does not admit such 
a {\em partially} symmetric stable matching, then we show that $G$ has {\em no} fully popular matching.

\subsection{Two partitions of the vertex set}
\label{sec:subsets}
We run the Gale-Shapley algorithm that avoids all forbidden edges~\cite{GI89} in $H$. In this 
algorithm, vertices on the left of $H$ propose in decreasing order of preference and vertices on the right of $H$ dispose.
When $u_{\ell} \in A_L \cup B_L$ proposes to $v_r^-$, this proposal is made along $(u_{\ell}^+,v_r^-)$: so $v_r$ sees this as $u_{\ell}^+$'s 
proposal; when $u_{\ell}$ proposes to $v_r^+$, this proposal is made along $(u_{\ell}^-,v_r^+)$: so $v_r$ sees this as $u_{\ell}^-$'s 
proposal. 

If $u_{\ell}$ proposes to a neighbor $v_r$ along $(u_{\ell}^+,v^-_r)$ or $(u_{\ell}^-,v^+_r)$, 
then $v_r$ (tentatively) accepts $u_{\ell}$'s proposal only if the edge $(u,v)$ is legal; otherwise $v_r$ rejects $u_{\ell}$'s proposal 
since this is a forbidden edge. Edges ranked worse than $(u_{\ell}^+,v^-_r)$/$(u_{\ell}^-,v^+_r)$ (as the case may be)
will be deleted from the current instance---this ensures that once $v_r$ receives a proposal 
along a certain edge, whether this proposal is (tentatively) accepted or not,  $v_r$ cannot accept proposals made along 
worse edges.

Let $S_0$ be the legal stable matching in $H$ that is obtained. The sets $U_A$ and $U_B$ will be useful.

\begin{itemize}
\item Let $U_A \subseteq A$ be the set of agents $a$ such that $(a^-_{\ell},a^+_r) \in S_0$. 
\item Let $U_B \subseteq B$ be the set of jobs $b$ such that $(b^-_{\ell},b^+_r) \in S_0$. 
\end{itemize}

We know that $H$ is made up of two halves: the upper half and the lower half.  Since $S_0$ is stable and thus perfect (recall that any
stable matching in $H$ is perfect), the set of agents matched to {\em genuine} neighbors---not to their twins---in  each half of $H$ 
is $A \setminus U_A$ and similarly, the set of jobs matched to genuine neighbors in each half of $H$ is $B \setminus U_B$. 

We now define the sets $A_+,A_{-},B_+,B_{-},A'_+,A'_{-},B'_+,B'_{-}$. 
Initially these sets are empty. Then we add vertices to them as described below. 

\begin{itemize}
\item For every $(a^+_{\ell},b^-_r) \in S_0$ where $a \in A$ and $b \in B$: add $a$ to $A_+$ and $b$ to $B_{-}$.
\item For every $(a^-_{\ell},b^+_r) \in S_0$ where $a \in A$ and $b \in B$: add $a$ to $A_{-}$ and $b$ to $B_+$.

\item For every $(b^+_{\ell},a^-_r) \in S_0$ where $a \in A$ and $b \in B$: add $b$ to $B'_+$ and $a$ to $A'_{-}$.
\item For every $(b^-_{\ell},a^+_r) \in S_0$ where $a \in A$ and $b \in B$: add $b$ to $B'_{-}$ and $a$ to $A'_+$.
\end{itemize}

\begin{figure}[h]
\centerline{\resizebox{0.6\textwidth}{!}{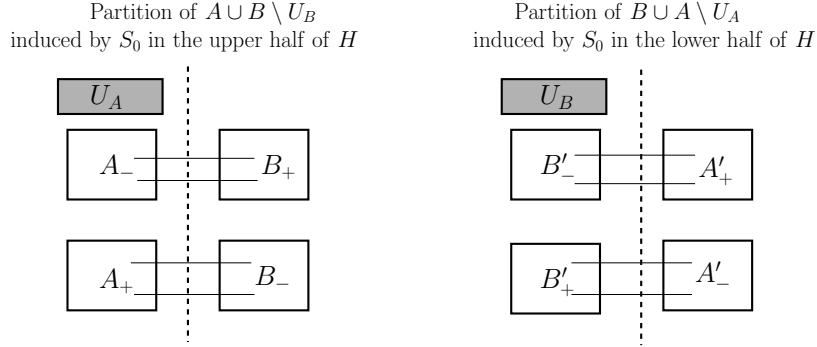}}
\caption{The two partitions of the set $(A \cup B) \setminus (U_A\cup U_B)$ induced by the matching $S_0$.}
\label{fig1}
\end{figure}

We have $A\setminus U_A = A_+ \cup A_{-} = A'_+ \cup A'_{-}$ and $B\setminus U_B = B_+ \cup B_{-} = B'_+ \cup B'_{-}$.
Fig.~\ref{fig1} denotes these partitions of $A\setminus U_A$ and $B\setminus U_B$ induced by the matching $S_0$ in the upper and 
lower halves of  $H$. In this figure, the set $U_A$ has been included in the upper half and the set $U_B$ in the lower half. 

We will use $(a^-_{\ell},\ast)$ to denote any edge in the set $\{(a^-_{\ell},b^+_r): b \in B\} \cup \{(a^-_{\ell},a^+_r)\}$. 
Similarly $(\ast,a^+_r)$ denotes any edge in the set $\{(b^-_{\ell},a^+_r): b \in B\} \cup \{(a^-_{\ell},a^+_r)\}$.
Similarly for $(b^-_{\ell},\ast)$ and $(\ast,b^+_r)$. Recall that every popular matching in $G$ has a {\em witness}
$\vec{\alpha} \in \{0, \pm 1\}^n$ (see Theorem~\ref{thm:witness}).

The following lemma will be crucial to us.

\begin{lemma}
\label{lemma:forced-agents}
Let $N$ be a fully popular matching in $G$ and let  $\vec{\alpha}$ be any
witness of $N$. If $a \in A_{-} \cap A'_+$ then $\alpha_a = 0$.
\end{lemma}
\begin{proof}
The vertex $a \in A_{-} \cap A'_+$, where the sets $A_{-}$ and $A'_+$ are defined above.
Let ${\cal D}_0$ be the set of legal stable matchings in $H$. The set ${\cal D}_0$ forms a sublattice of the lattice\footnote{The meet of 2 stable matchings $M$ and $M'$ is the stable matching where every $u$ in $A_L\cup B_L$ (resp., $A_R\cup B_R$) is matched to its {\em more} (resp., {\em less}) preferred partner in $\{M(u),M'(u)\}$. The join of $M$ and $M'$ is the stable matching where every $u$ in $A_L\cup B_L$ (resp., $A_R\cup B_R$) is matched to its {\em less} (resp., {\em more}) preferred partner in $\{M(u),M'(u)\}$.} of stable matchings
in $H$ and the matching $S_0$ is the {\em $(A_L \cup B_L)$-optimal} matching in ${\cal D}_0$~\cite{GI89}. 
Since $a \in A_{-}$, we have $(a^-_{\ell},c^+_r) \in S_0$ for some neighbor $c^+_r$ of $a_{\ell}$.
Thus $c^+_r$ is the most preferred partner for $a_{\ell} \in A_L$ in all matchings in ${\cal D}_0$.
Recall that every vertex prefers superscript~$-$ neighbors to superscript~$+$ neighbors.
Hence no matching in ${\cal D}_0$ matches $a_{\ell}$ to a superscript~$-$ neighbor, i.e.,
every legal stable matching in $H$ has to contain $(a^-_{\ell},\ast)$.

$S_0$ is also the {\em $(A_R \cup B_R)$-pessimal} matching in ${\cal D}_0$~\cite{GI89}.
Since $a \in A'_+$, we have $(d_{\ell}^-,a_r^+) \in S_0$ for some neighbor $d^-_{\ell}$ of $a_r$.
So every matching in ${\cal D}_0$ has 
to match $a_r \in A_R$ to a neighbor at least as good as $d_{\ell}^-$, i.e., every legal stable matching in $H$ 
has to contain $(\ast,a_r^+)$. 

Suppose $N$ is a legal stable matching with a witness $\vec{\alpha}$ such that $\alpha_a \in \{\pm 1\}$.
If $\alpha_a = 1$, i.e., if $a \in A_1$ (see Fig.~\ref{fig:first}),
then there is a legal stable matching $N^*_{\alpha}$ in $H$ such that $(a_{\ell}^+,\ast) \in N^*_{\alpha}$ 
(see the proof of Lemma~\ref{lem:main}).
This contradicts our claim above that every legal stable matching in $H$ has to contain $(a_{\ell}^-,\ast)$.
So $\alpha_a = -1$, i.e., $a \in A_{-1}$. 
Then there is a legal stable matching $N^*_{\alpha}$ in $H$ such that $(\ast,a_r^-) \in N^*_{\alpha}$ (as shown in 
the proof of Lemma~\ref{lem:main}).
This again contradicts our claim above that every legal stable matching in $H$ has to contain $(\ast,a_r^+)$.
Thus $\alpha_a \notin \{\pm 1\}$, hence  $\alpha_a = 0$. \qed
\end{proof}

\begin{lemma}
  \label{lemma:forced-jobs}
  Let $N$ be a fully popular matching in $G$ and let  $\vec{\alpha}$ be any
witness of $N$. If $b \in B_+ \cap B'_{-}$ then $\alpha_b = 0$.
\end{lemma}

The proof of Lemma~\ref{lemma:forced-jobs} is analogous to the proof of Lemma~\ref{lemma:forced-agents}. 
We will use $G_0 = (A\cup B,E_0)$ to denote the {\em popular subgraph} of $G = (A \cup B, E')$. The edge set $E_0$ 
of $G_0$ is the set of popular edges/self-loops.
The subgraph $G_0$ need not be connected and Lemma~\ref{lemma1} will be useful to us.

\begin{lemma}[\cite{FK20}]
\label{lemma1}
Let $C$ be any connected component in the popular subgraph $G_0$. For any popular matching $N$ in $G$ and any witness 
$\vec{\alpha}$ of $N$: if $\alpha_v = 0$ for some $v \in C$ then $\alpha_u = 0$ for all $u \in C$.
\end{lemma}
\begin{proof}
 Consider any popular edge $(a,b)$. So there is some popular matching 
 $M$ with the edge $(a,b)$. The matching $M$ is an optimal solution to the max-weight perfect matching LP with edge weight
 function $\wt_N$ since $\wt_N(M) = \phi(M,N) - \phi(N,M) = 0$: recall that $M$ and $N$ are popular matchings in $G$.
 We know that  $\vec{\alpha}$ is an optimal solution to the dual LP. So it follows from complementary slackness
 that $\alpha_a + \alpha_b = \wt_N(a,b)$. Since $\wt_N(a,b) \in \{\pm 2, 0\}$ (an even number), the integers $\alpha_a$ and
 $\alpha_b$ have the same parity.

 Let $u$ and $v$ be any two vertices in the same connected component in the popular subgraph $G_0$. So there is a $u$-$v$ path
 $\rho$ in $G$ such that every edge in $\rho$ is a popular edge. We have just seen that the endpoints of each popular edge
 have the same parity in $\vec{\alpha}$. Hence $\alpha_u$ and $\alpha_v$ have the same parity. Thus $\alpha_v = 0$ implies
 $\alpha_u = 0$. \qed 
\end{proof}

\subsection{Our algorithm}
\label{sec:algorithm}

Lemmas~\ref{lemma:forced-agents}-\ref{lemma1} motivate our algorithm which is described as Algorithm~\ref{alg:fully-popular}.
The main step of the algorithm is the {\em while} loop that takes any unmarked vertex $v$ in $(A_{-} \cap A'_+) \cup (B_+ \cap B'_{-})$.
Initially all vertices are unmarked. Consider the first iteration of the algorithm: let $v \in A$.

Lemma~\ref{lemma:forced-agents} tells us that for any fully popular matching $N$ and any witness $\vec{\alpha}$ of $N$, we have
$\alpha_v = 0$. Lemma~\ref{lemma1} tells us that $\alpha_u = 0$ for every vertex $u$ in the component $C$, where $C$ is 
$v$'s connected component in $G_0$. The proof of Lemma~\ref{lem:main} shows $N$ has a realization $N^*_{\alpha}$ in $H$ such that 
$N^*_{\alpha}$ contains $(a_{\ell}^-,\ast)$ and $(\ast,a_r^+)$ for every agent $a \in C$. 

Thus we are interested in those legal stable matchings in $H$ that contain $(a_{\ell}^-,\ast)$ and $(\ast,a_r^+)$ 
for every agent $a \in C$. Hence our algorithm {\em forbids} all edges $(a_{\ell}^+,\ast)$ and 
$(\ast,a_r^-)$ for every agent $a \in C$ in the stable matching that we compute here.
This step is implemented by making every neighbor reject offers from $a_{\ell}^+$ (this may induce other rejections) 
and symmetrically, $a_r$ rejects all offers from superscript~$+$ neighbors. Note that the resulting matching may contain 
$(a_{\ell}^-,a_r^+)$ for some of the agents $a$ in $C$. All the vertices in $C$ get marked in this iteration.

\begin{algorithm}
\caption{Our algorithm to find a fully popular matching in $G = (A \cup B, E')$}
\label{alg:fully-popular}
\begin{algorithmic}[1]
\State Compute a legal stable matching $S_0$ in $H$ by running the Gale-Shapley algorithm with forbidden edges.

\Comment{{\em (Vertices in $A_L \cup B_L$ propose and those in $B_R \cup A_R$ dispose. Any edge 

\hspace*{4.9cm} in $H$ whose corresponding edge in $G$ is not legal is forbidden.)}}
\State Let $A_{-},A'_+$ and $B_+,B'_{-}$ be as defined earlier (see the start of Section~\ref{sec:subsets}).
\State Initially all vertices are unmarked and $i = 0$.
\While{there exists an unmarked vertex $v \in (A_{-} \cap A'_+) \cup (B_+ \cap B'_{-})$}
  \State $i = i+1$.
  \State Modify $S_{i-1}$ to $S_i$ so as to forbid all edges $(a_{\ell}^+,\ast)$ and $(\ast,a_r^-)$ for every agent $a$ in $v$'s component in the \hspace*{0.47cm} popular subgraph $G_0$.
  
  \Comment{{\em ($S_i$ is the $(A_L\cup B_L)$-optimal legal stable matching in $H$ that avoids all forbidden edges identified in the \hspace*{1.1cm} first $i$ iterations of the while-loop.)}}
  \If{there is no such legal stable matching $S_i$ in $H$} 
      \State Return ``No fully popular matching in $G$''.
  \EndIf 
  \State Update the sets $A_{-},A'_+$ and $B_+,B'_{-}$: these correspond to $S_i$ now. 
  \State Mark all vertices in $v$'s component in the popular subgraph $G_0$.   
\EndWhile
\State Return $M = \{(a,b) \in E: (a^+_{\ell},b^-_r)\ \mathrm{or}\ (a^-_{\ell},b^+_r)\ \mathrm{is\ in}\ S_i\}$.
\end{algorithmic}
\end{algorithm}

Recall that ${\cal D}_0$ is the set of legal stable matchings in $H$.
Let ${\cal D}_1 \subseteq {\cal D}_0$ be the set of all legal stable matchings in $H$ that contain $(a_{\ell}^-,\ast)$ and 
$(\ast,a_r^+)$ for every agent $a \in C$. Thus ${\cal D}_1$ is a sublattice of ${\cal D}_0$.
We know from the proof of Lemma~\ref{lem:main} that $N^*_{\alpha} \in {\cal D}_1$ where
$N$ is a fully popular matching in~$G$ and $\vec{\alpha}$ is any witness of $N$.
So if ${\cal D}_1$ is empty then we can conclude that $G$ has {\em no} fully popular matching. 
Otherwise, we have a matching $S_1 \in {\cal D}_1$ with us and we update the sets $A_{-},A'_+$ and $B_+,B'_{-}$: these
sets are defined at the start of Section~\ref{sec:subsets} and now $S_1$ replaces $S_0$ in their definitions. 

Let us assume we are now in the $i$-th iteration and let ${\cal D}_i$ be the set of legal stable matchings in $H$ that avoid
all edges forbidden by our algorithm in the first $i$ iterations. In other words, ${\cal D}_i$ is the set of those matchings in 
${\cal D}_{i-1}$ where no edge identified as {\em forbidden} in the $i$-th iteration is present. 
We have ${\cal D}_0 \supseteq {\cal D}_1 \supseteq \cdots \supseteq {\cal D}_{i-1} \supseteq {\cal D}_i$.
For all $0 \le j \le i$, the set ${\cal D}_j$ forms a sublattice of the lattice of all stable matchings in $H$~\cite{GI89}. 

\begin{lemma}
\label{lemma:new}
For every fully popular matching $N$ in $G$ and every witness $\vec{\alpha}$ of $N$, the realization $N^*_{\alpha}$ 
is an element of ${\cal D}_i$.
\end{lemma}

\begin{proof}
We need to show that $N^*_{\alpha}$ is an element of ${\cal D}_i$.
We will prove this by induction. We know from Lemma~\ref{lem:main} that the base case is true, i.e., $N^*_{\alpha} \in {\cal D}_0$.
By induction hypothesis, let us assume that for every fully popular matching $N$ and any witness $\vec{\alpha}$ of $N$, the realization 
$N^*_{\alpha}$ is an element of 
${\cal D}_{i-1}$. Since the algorithm entered the $i$-th iteration of the while loop, there was an unmarked vertex $x$ in 
$(A_{-} \cap A'_+) \cup (B_+ \cap B'_{-})$ at the start of this iteration. 

\begin{new-claim}
\label{clm5}
For any fully popular matching $N$ and any witness $\vec{\alpha}$ of $N$, we have $\alpha_x = 0$.
\end{new-claim}

The proof of Claim~\ref{clm5} (this is similar to the proof of Lemma~\ref{lemma:forced-agents}) is given below. Claim~\ref{clm5} along with 
Lemma~\ref{lemma1} tells us that for all vertices $u$ in $x$'s component $C'$ in $G_0$, 
we have $\alpha_u = 0$. The proof of Lemma~\ref{lem:main} shows us that $N^*_{\alpha}$ contains $(a_{\ell}^-,\ast)$ and $(\ast,a_r^+)$ for every 
agent $a \in C'$. Since $N^*_{\alpha} \in {\cal D}_{i-1}$, it follows that $N^*_{\alpha}$ is an element in ${\cal D}_i$. 
Thus for every fully popular matching $N$ in $G$ and every witness $\vec{\alpha}$ of~$N$, the realization $N^*_{\alpha}$ 
is an element of ${\cal D}_i$. \qed
\end{proof}

\paragraph{\bf Proof of Claim~\ref{clm5}.}
The matching $S_{i-1}$ that is computed in line~6 of the $(i-1)$-th iteration is the $(A_L\cup B_L)$-optimal matching in the
lattice ${\cal D}_{i-1}$~\cite{GI89}. Hence if $(x_{\ell}^-,\ast) \in S_{i-1}$ for some $x_{\ell} \in A_L\cup B_L$ then $(x_{\ell}^-,\ast)$
belongs to every matching in ${\cal D}_{i-1}$. The matching $S_{i-1}$ is also the $(A_R \cup B_R)$-pessimal matching in the set
${\cal D}_{i-1}$~\cite{GI89}. Hence if 
$(\ast,x_r^+) \in S_{i-1}$ for some $x_r \in A_R\cup B_R$ then $(\ast,x_r^+)$  belongs to every matching in ${\cal D}_{i-1}$.

If the above claim 
is false then there is a fully popular matching $N$ and a witness $\vec{\alpha}$ of $N$ with
$\alpha_x \in \{\pm 1\}$. If $\alpha_x = 1$ then there is a legal stable matching $N^*_{\alpha}$ in $H$ such that 
$(x_{\ell}^+,\ast) \in N^*_{\alpha}$.
If $\alpha_x = -1$ then there is a legal stable matching $N^*_{\alpha}$ in $H$ such that $(\ast,x^-_r) \in N^*_{\alpha}$.
Since $N^*_{\alpha} \in {\cal D}_{i-1}$, both cases contradict our earlier observation that every matching in 
${\cal D}_{i-1}$ has to contain $(x_{\ell}^-,\ast)$ and $(\ast,x^+_r)$. Thus for any fully popular matching $N$ and any witness 
$\vec{\alpha}$ of $N$, we have $\alpha_x = 0$. \qed

\medskip

We now need to prove the correctness of our algorithm. Suppose the algorithm returns ``No fully popular matching in $G$''. Then this means that ${\cal D}_i = \emptyset$ for some $i \ge 1$.
Lemma~\ref{lemma:new} tells us that if ${\cal D}_i = \emptyset$, then there is indeed no fully popular matching in $G$.
This finishes one part of our proof of correctness. 

Suppose the algorithm does not return ``No fully popular matching in $G$''.
Since at least one unmarked vertex gets marked in every iteration of the while loop, the algorithm always terminates. 
So a matching~$M$ is returned. We need to show that $M$ is a fully popular matching in $G$. 
This is the tougher side in the proof of correctness and this is proved in Section~\ref{sec:correct}.

\section{Popularity of the Matching $M$}
\label{sec:correct}

In this section we complete the proof of correctness of Algorithm~\ref{alg:fully-popular}.
We need to show that the matching $M$ returned by Algorithm~\ref{alg:fully-popular} is fully popular in $G$.
Let $S_i$ be the matching in $H$ computed in the final iteration of Algorithm~\ref{alg:fully-popular}.
Then $M$ is the matching (in $G$) induced by $S_i$ in the upper half of $H$. 
The matching $M$ is as defined below:
\[M = \{(a,b) \in E: (a^+_{\ell},b^-_r)\ \mathrm{or}\ (a^-_{\ell},b^+_r)\ \mathrm{is\ in}\ S_i\}.\]

Note that $M \subseteq (A_+ \times B_{-}) \cup (A_{-} \times B_+)$, where the sets $A_+,B_{-},A_{-},B_+$ are defined at the beginning
of Section~\ref{sec:subsets}: the matching $S_i$ replaces $S_0$ in the definitions of $A_+,B_{-},A_{-},B_+,A'_+,B'_{-},A'_{-},B'_+$ now.
Similarly, let $L$ be the matching (in $G$) induced by $S_i$ in the lower half of $H$. So we have:
\[L = \{(a,b) \in E: (b^+_{\ell},a^-_r)\ \mathrm{or}\ (b^-_{\ell},a^+_r)\ \mathrm{is\ in}\ S_i\}.\]

Thus $L \subseteq (A'_+ \times B'_{-}) \cup (A'_{-} \times B'_+)$.
Let $U_A$ (resp., $U_B$) be the set of vertices $u$ in $A$ (resp., $B$) such that $(u^-_{\ell},u^+_r) \in S_i$. 
The vertices in $U_A \cup U_B$ are unmatched in both $M$ and $L$. 

Since $S_i$ is a legal stable matching in $H$, it matches all vertices in $H$ using valid edges.
Thus by Theorem~\ref{thm:A-popular}, $M$ is $A$-popular.\footnote{In order to apply Theorem~\ref{thm:A-popular}, we ought to say $M \cup \{(u,u): u \in U_A\cup U_B\}$ is $A$-popular.}
We need to show that $M$ is popular in $G$. 

Theorem~\ref{thm:M-U-M-L} is our starting point.
The subgraph $G \setminus U_B$ is the subgraph of $G$ induced on $A\cup (B\setminus U_B)$ and similarly,
the subgraph $G \setminus U_A$ is the subgraph of $G$ induced on $(A\setminus U_A)\cup B$.

\begin{theorem}
\label{thm:M-U-M-L}
The matching $M$ is popular in the subgraph $G \setminus U_B$. Also, the matching $L$ is popular in the subgraph $G \setminus U_A$.
\end{theorem}
\begin{proof}
We will use Theorem~\ref{thm:witness} to prove the popularity of $L$ and $M$ in  $G \setminus U_A$ and $G \setminus U_B$,
respectively.
The popularity of $L$ in  $G \setminus U_A$ will be shown using the witness $\vec{\beta}$ defined below and the
popularity of $M$ in  $G \setminus U_B$ will be shown using the witness $\vec{\gamma}$ defined below.

\begin{enumerate}
\item $\beta_u = 1$ for $u \in A'_+ \cup B'_+$, \ \ \  $\beta_u = -1$ for $u \in A'_{-} \cup B'_{-}$, \ \ \ and\ \ \ $\beta_u = 0$ for $u \in U_B$.
\item $\gamma_u = 1$ for $u \in A_+ \cup B_+$, \ \ \  $\gamma_u = -1$ for $u \in A_{-} \cup B_{-}$, \ \ \ and \ \ $\gamma_u = 0$ for $u \in U_A$.
\end{enumerate}

Observe that $\sum_{u \in (A\setminus U_A) \cup B}\beta_u = 0$. This is because $L \subseteq (A'_+ \times B'_{-}) \cup (A'_{-} \times B'_+)$.
Note that $\wt_{L}(u,u) = 0$ for $u \in U_B$ and $\wt_{L}(u,u) = -1$ for all $u \notin U_A\cup U_B$. Thus
we have $\beta_u \ge \wt_{L}(u,u)$ for all $u \in (A\setminus U_A) \cup B$.

Similarly, $\sum_{u \in A \cup (B\setminus U_B)}\gamma_u = 0$.
Also, $\gamma_u \ge \wt_{M}(u,u)$ for all $u \in A \cup (B\setminus U_B)$.

\begin{new-claim}
  \label{claim16}
  $\beta_a + \beta_b \ge \wt_{L}(a,b)$ for all edges $(a,b)$ where $a \in A \setminus U_A$ and $b \in B$. 
\end{new-claim}

\begin{new-claim}
  \label{claim17}
  $\gamma_a + \gamma_b \ge \wt_M(a,b)$ for all edges $(a,b)$ where $a \in A$ and $b \in B \setminus U_B$. 
\end{new-claim}
We will prove Claim~\ref{claim17} below. The proof of Claim~\ref{claim16} is analogous.
\begin{itemize}
\item Case~1: let $a \in U_A$. 
We set $\gamma_a = 0$ and we know that $(a_{\ell}^-,a_r^+) \in S_i$.
Recall that $a_r^+$ is $a_{\ell}$'s least preferred neighbor, thus $a_{\ell}$ must have been rejected by all its more preferred neighbors. 
That is, every neighbor $b^+_r$ of $a_{\ell}$ must have received a proposal from $a_{\ell}^-$. 
Since $b_r$ prefers superscript~$-$ neighbors to superscript~$+$ neighbors, this means $(d^-_{\ell},b^+_r) \in S_i$ for some neighbor 
$d^-_{\ell}$ that $b_r$ prefers to $a^-_{\ell}$, i.e., $b$ prefers $d$ to $a$. 
Thus $b \in B_1$ (so $\gamma_b = 1$) and moreover, $\wt_{M}(a,b) = 0$. Hence 
$\gamma_a + \gamma_b = 1 > \wt_{M}(a,b)$.

\medskip

\item Case~2: let $a \in A_{-}$. 
There are two possibilities: (1)~$b \in B_{-}$ and 
(2)~$b \in B_+$. Suppose $b \in B_-$. Then we have $(a^-_{\ell},c^+_r)$ and $(d^+_{\ell},b^-_r)$ in $S_i$ for some neighbors 
$c$ and $d$ of $a$ and $b$, respectively. Since every vertex prefers superscript~$-$ 
neighbors to superscript~$+$ neighbors, it means $a_{\ell}$ proposed to $b^-_r$ and got rejected, i.e., $b_r$ prefers its
partner $d^+_{\ell}$
to $a^+_{\ell}$. We also claim $a_{\ell}$ prefers its partner $c^+_r$ to $b^+_r$. This is because $b_r$ prefers $a^-_{\ell}$ 
to $d^+_{\ell}$ 
(superscript~$-$ neighbors over superscript $+$ neighbors): so if $a^-_{\ell}$ had proposed to $b_r$, then $b_r$ would have
rejected its partner $d^+_{\ell}$. This means that both $a$ and $b$ prefer their partners in $M$ to each other. Thus
$\wt_{M}(a,b) = -2 = \gamma_a + \gamma_b$.

\smallskip

Suppose $b \in B_+$. Then either (i)~$(a^-_{\ell},b^+_r) \in S_i$ or (ii)~$(a^-_{\ell},c^+_r)$ and $(d^-_{\ell},b^+_r)$ are 
in $S_i$ for some neighbors $c$ and $d$ of $a$ and $b$, respectively.
In subcase~(i), we have $\wt_{M}(a,b) = 0 = \gamma_a + \gamma_b$ and in subcase~(ii), the stability of $S_i$ in $H$ implies that either
$a_{\ell}$ prefers $c^+_r$ to $b^+_r$ or $b_r$ prefers $d^-_{\ell}$ to $a^-_{\ell}$, thus $\wt_{M}(a,b) \le 0 = \gamma_a + \gamma_b$.

\medskip

\item Case~3: let $a \in A_{+}$. 
As before, there are two possibilities: $b \in B_+$ and $b \in B_{-}$. When $b \in B_+$, we have $\gamma_a + \gamma_b = 2$ 
and since $\wt_{M}(a,b) \le 2$, the constraint $\wt_{M}(a,b) \le \gamma_a + \gamma_b$ obviously holds.

When $b \in B_-$, either (i)~$(a^+_{\ell},b^-_r) \in S_i$ or (ii)~$(a^+_{\ell},c^-_r)$ and $(d^+_{\ell},b^-_r)$ are in $S_i$.
In subcase~(i), we have $\wt_{M}(a,b) = 0 = \gamma_a + \gamma_b$ and in subcase~(ii), it follows from the stability of $S_i$ in $H$
that either $a_{\ell}$ prefers $c^-_r$ to $b^-_r$ or $b_r$ prefers $d^+_{\ell}$ to $a^+_{\ell}$.
Thus $\wt_{M}(a,b) \le 0 = \gamma_a + \gamma_b$. 
\end{itemize}

This finishes the proof of $M$'s popularity in $G \setminus U_B$ (by Theorem~\ref{thm:witness}). 
Similarly, $L$ is popular in $G \setminus U_A$ by using the witness $\vec{\beta}$ defined above. \qed 
\end{proof}

Theorem~\ref{thm:M-U-M-L} tells us that the matching $M$ is popular in the subgraph $G \setminus U_B$.
However we need to prove the popularity of $M$ in the {\em entire} graph $G$, i.e., we need to include vertices in
$U_B$ as well. Setting $\gamma_b = 0$ for $b \in U_B$ will not cover edges in $A_{-}\times U_B$. 
To prove the popularity of $M$ in $G$, we will use the fact that $L$ is popular in $G \setminus U_A$ and show that $M$ and $L$
have several edges in common.

Let $Z$ be the set of all vertices outside $U_A \cup U_B$ that got marked in our algorithm. So these are the marked vertices
that are matched in $S_i$ to genuine neighbors (not to their twins). 
Since we marked entire connected components in the popular subgraph $G_0$ 
in Algorithm~1, both $M$ and $L$ match vertices in $Z$ to each other.

Lemma~\ref{new-correct2} shows that the matching $S_i$ has ``partial symmetry'' across the upper and lower halves of the graph $H$; more
precisely, $M$ and $L$ are identical on the set $Z$. This will be key to showing $M$'s popularity in $G$.
The following lemma will be useful in proving Lemma~\ref{new-correct2}.

\begin{lemma}
\label{correct0}
$M$ and $L$ are stable matchings when restricted to vertices in $Z \cup U_A \cup U_B$.
\end{lemma}
\begin{proof}
Let $Z_A = Z \cap A$ and let $Z_B = Z \cap B$. It follows from our algorithm that $Z_A \subseteq A_{-} \cap A'_+$
and $Z_B \subseteq B_+ \cap B'_{-}$ (see Fig.~\ref{fig:third}). 

\begin{figure}[t]
\centerline{\resizebox{0.6\textwidth}{!}{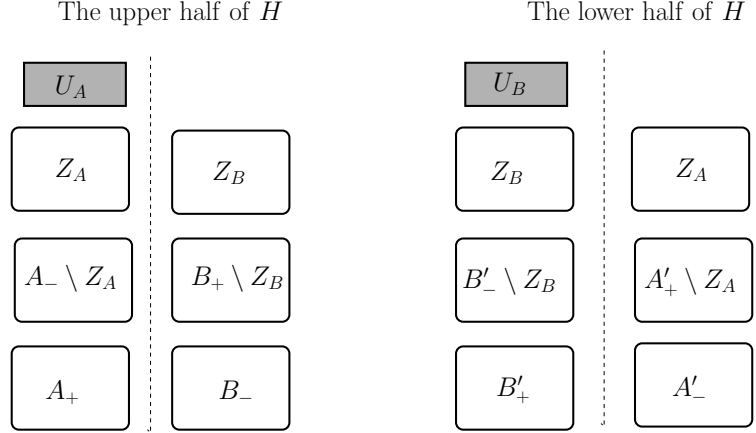}}
\caption{The final picture of the partitions created by $M$ and $L$ in the upper and lower halves of $H$, respectively. The while loop 
termination condition implies $(A_{-}\setminus Z_A) \subseteq A'_{-}$ and $(A'_+\setminus Z_A) \subseteq A_+$ and so on.}
\label{fig:third}
\end{figure}

We need to show that $M$ (similarly, $L$) has no blocking edge in $(Z_A \cup U_A) \times (Z_B \cup U_B)$. 
Consider any edge $(a,b) \in Z_A \times Z_B$. We have $\gamma_a = -1$ and $\gamma_b =1$ while  $\beta_a = 1$ and $\beta_b = -1$. 
We know from Claim~\ref{claim17} (given in the proof of Theorem~\ref{thm:M-U-M-L}) that 
$\wt_{M}(a,b) \le \gamma_a + \gamma_b = -1 + 1 = 0$. Similarly, $\wt_{L}(a,b) \le \beta_a + \beta_b = 1-1 = 0$. 
Thus $(a,b)$ is not a blocking edge to either $M$ or $L$.
Hence neither $M$ nor $L$ has a blocking edge in $Z_A \times Z_B$. 

Moreover, $G$ has no edge in $U_A \times U_B$. This is because
each vertex $u \in U_A \cup U_B$ has to be an unstable vertex---otherwise $u$ is stable and so $(u^-_{\ell},u^+_r)$ is an unpopular edge 
and thus forbidden.

Consider any edge $(a,b) \in U_A \times Z_B$. We have $\gamma_a = 0$ and $\gamma_b =1$. So $\wt_M(a,b) \le \gamma_a + \gamma_b = 0 + 1 = 1$. 
Since $\wt_M(a,b)$ is an even number, this means $\wt_M(a,b) \le 0$. Thus $(a,b)$ is not a blocking edge to $M$. 
We will next show that $(a,b)$ is not a blocking edge to $L$.

Since $a \in U_A$ and $b \in Z_B\subseteq B'_{-}$, the edges $(a^-_{\ell},a^+_r)$ and  $(b^-_{\ell},c^+_r)$ are
in $S_i$ for some neighbor $c$ of $b$. Note that $a^-_{\ell}$ is $a^+_r$'s least preferred superscript
$-$ neighbor in $H$. Thus $a^+_r$ did not receive any offer from $b^-_{\ell}$ in Algorithm~1.  Because $S_i$ is stable in $H$, 
it has to be the case that $b_{\ell}$ prefers $c^+_r$ to $a^+_r$.
Since $(c,b) \in L$, we have $\wt_L(a,b) = 0$. Thus $(a,b)$ is not a blocking edge to $L$.

An analogous argument shows that no edge in $Z_A \times U_B$ blocks either $M$ or $L$.
Thus $M$ and $L$ are stable matchings when restricted to vertices in $Z \cup U_A \cup U_B$. \qed
\end{proof}

\begin{lemma}
\label{new-correct2}
The matching $M$ restricted to vertices in $Z$ is the same as the matching $L$ restricted to vertices in $Z$.
\end{lemma}
\begin{proof}
Consider any connected component $C$ in the popular subgraph $G_0$. The component $C$ splits into
sub-components $C'_1,\ldots,C'_t$ when we restrict edges to only those marked ``valid''. We claim there is exactly {\em one} stable 
matching $T_{C'_j}$ in each such sub-component $C'_j$. 
Assume $C'_j$ contains a job $b$ that is a top choice neighbor for some agent.\footnote{Otherwise $C'_j$ consists of a single edge $(a,s(a))$ for some $a \in A$; if there was another agent $a'$ in $C'_j$ then $s(a') = s(a)$ and so one of $a, a'$ would be left unmatched in $S_i$, a contradiction to $S_i$'s stability in $H$.} Then $b$ has to be matched in $T_{C'_j}$ to its 
most preferred neighbor $a$ in $C'_j$, otherwise $(a,b)$ would be a blocking edge to $T_{C'_j}$. Recall that every agent has exactly two 
valid edges incident to it. So fixing one edge $(a,b)$ in the matching fixes $T_{C'_j}$.

In more detail, every agent $a' \ne a$ in $C'_j$ such that $f(a') = b$ has to be matched in $T_{C'_j}$ to $s(a')$ (call it~$b'$).
Given that $a'$ is matched to $b'$, every agent $a'' \ne a'$ in $C'_j$ such that $s(a'') = b'$ has to be matched in $T_{C'_j}$ 
to $f(a'')$ and so on. Thus the matching 
$T_{C'_j}$ gets fixed. The same happens with every sub-component in $C$ and so the only stable matching in $C$ is $T_C = \cup_{j=1}^tT_{C'_j}$.

Let $C_1,\ldots,C_r$ be the connected components of $G_0$ that contain vertices in $Z$. So all vertices in $\cup_{i=1}^r C_i$ are marked, thus
$\cup_{i=1}^r C_i \subseteq Z \cup U_A \cup U_B$.
We know from Lemma~\ref{correct0} that both $M$ and $L$ are stable matchings in each $C_i$, where $1 \le i \le r$.
So $M$ (similarly, $L$) restricted to $\cup_{i=1}^r C_i$ is $\cup_{i=1}^rT_{C_i}$. Thus $M$ and $L$ have the same edges on $Z$. \qed
\end{proof}

Lemma~\ref{new-correct2} helps us in defining an appropriate witness $\vec{\alpha}$ to show $M$'s popularity in $G$.
Recall the vector $\vec{\gamma}$ defined in Theorem~\ref{thm:M-U-M-L}: we will set $\alpha_u = 0$ for all $u \in Z \cup U_B$ and 
$\alpha_u = \gamma_u$ otherwise. Before we use this vector $\vec{\alpha}$ to prove the popularity of $M$ in Theorem~\ref{thm:final}, 
we need the following two lemmas.

\begin{lemma}
\label{correct1}
For every $a \in A_{-} \setminus Z_A$, $a$ likes $M(a)$ at least as much as $L(a)$.
\end{lemma}
\begin{proof}
Suppose not. Then $M(a) = s(a)$ while $L(a) = f(a)$.
We claim $f(a) \in B_+$. Otherwise $f(a) \in B_{-}$, however for every edge 
$(x,y) \in A_{-}\times B_{-}$, we have $\wt_{M}(x,y) \le \gamma_x + \gamma_y = -2$ (by Claim~\ref{claim17}). 
But $a$ prefers $f(a)$ to its partner in $M$, thus $\wt_{M}(a,f(a)) \ge 0$. Hence $f(a) \in B_+$.
Since $\wt_{M}(x,y) \le 0$ for every edge $(x,y) \in A_{-} \times B_+$, we can conclude that $\wt_{M}(a,f(a)) = 0$, i.e.,
$f(a)$ is matched in $M$ to a neighbor $a' \in A_{-}$ that it prefers to $a$. Since $S_i$ uses
only valid edges, this means $f(a) = f(a')$, i.e., $f(a)$ is the top choice neighbor of $a'$.

We now move to the lower half of $H$: observe that both $a$ and $a'$ are in $A'_{-}$. This is because there is no unmarked vertex 
in $A_{-} \cap A'_+$ by the termination condition of our while-loop. Note that $a$ is unmarked since $a \notin Z_A$. Thus $a'$ is also
unmarked since $(a, f(a))$ and $(a',f(a))$ are popular edges, hence $a$ and $a'$ are in the same connected component in $G_0$.
Since $a \in A'_{-}$, $L(a) = f(a)$ is in $B'_+$. Consider the edge $(a',f(a)) \in A'_{-} \times B'_+$: both $a'$ and $f(a)$ prefer each 
other to their respective partners in $L$. This means $\wt_{L}(a',f(a)) = 2$. However for each edge $(x,y) \in A'_{-} \times B'_+$, we have 
$\wt_{L}(x,y) \le \beta_x + \beta_y = 0$ (by Claim~\ref{claim16}), a contradiction. 
So for every $a \in A_{-} \setminus Z_A$, it has to be the case that $a$ likes $M(a)$ at least as much as $L(a)$. \qed
\end{proof}

\begin{lemma}
\label{correct3}
For every $a \in A_+ \cap A'_+$,  $a$ likes $M(a)$ at least as much as $L(a)$.
\end{lemma}
\begin{proof}
Suppose not. Then $M(a) = s(a)$ while $L(a) = f(a)$. Since $a \in A'_+$, $L(a) = f(a) \in B'_{-}$. This implies $f(a) \in B_{-}$ since 
there is no unmarked vertex in $B_+ \cap B'_{-}$ by the termination condition of our while-loop. We know $f(a)$ is unmarked
since $a$ (its partner in $L$) is unmarked and this is because $a \in A_+$. Since $a \in A_+$ and $f(a) \in B_{-}$, we have 
$\wt_{M}(a,f(a)) \le \gamma_a + \gamma_b = 0$ (by Claim~\ref{claim17}). So $f(a)$ has to be matched in $M$ to a more preferred
neighbor $a' \in A_+$.  As argued in the proof of Lemma~\ref{correct1}, it follows from the legality of $S_i$ that $f(a)$ is 
the top choice neighbor of $a'$.

Consider the matching $L$ in the lower half of $H$. Since $L(a) = f(a)$, $\wt_{L}(a',f(a))~=~2$.
That is, $(a',f(a))$ is a blocking edge to $L$. We need $\beta_{a'} = \beta_{f(a)} = 1$ to ensure 
$\beta_{a'} + \beta_{f(a)} \ge \wt_{L}(a',f(a)) = 2$ (by Claim~\ref{claim16}).
However $f(a) \in B'_{-}$ since $a \in A'_+$. This means $\beta_{f(a)} = -1$, a contradiction. 
Thus for any $a \in A_+ \cap A'_+$, it follows that $a$ likes $M(a)$ at least as much as $L(a)$. \qed
\end{proof}

We are now ready to prove the popularity of $M$ in $G$.

\begin{theorem}
\label{thm:final}
The matching $M$ is popular in $G$.
\end{theorem}
\begin{proof}
The popularity of $M$ in $G$ will be shown using $\vec{\alpha}$ defined below: 
\begin{itemize}
\item $\alpha_u = 0$ for $u \in Z \cup U_A \cup U_B$.
\item $\alpha_u = 1$ for $u \in A_+ \cup (B_+\setminus Z_B)$. 
\item $\alpha_u = -1$ for $u \in B_{-} \cup (A_{-}\setminus Z_A)$.
\end{itemize}

We have $M \subseteq (A_+\times B_{-}) \cup (Z_A \times Z_B) \cup ((A_{-}\setminus Z_A) \times (B_+\setminus Z_B))$ (see Fig.~\ref{fig:third}).
Thus $\sum_{u \in A \cup B}\alpha_u = 0$. Also, $\alpha_u \ge \wt_{M}(u,u)$ for all vertices $u \in A \cup B$ since 
$\alpha_u = 0 = \wt_{M}(u,u)$ for $u \in U_A \cup U_B$ and $\alpha_u \ge -1 = \wt_{M}(u,u)$ for all other~$u$. 
To show $M$'s popularity using Theorem~\ref{thm:witness}, we need to prove that $\alpha_a + \alpha_b \ge \wt_{M}(a,b)$ for all edges $(a,b)$.

We will first show this constraint holds for edges incident to vertices in $U_B$.
For this, we will use the matching $L$.
It is easy to see that the neighborhood of $U_B$ is in $A'_+$ and also that each  $a \in A'_+$ prefers its partner in $L$ to 
$b \in U_B$. This is because $(b_{\ell}^-,b_r^+) \in S_i$ and $b_r^+$ is $b_{\ell}$'s least preferred neighbor, thus $b_{\ell}$
must have been rejected by all its more preferred neighbors in our algorithm, i.e., every neighbor $a^+_r$ of $b_{\ell}$
received a proposal from $b_{\ell}^-$. Since $a_r$ prefers superscript~$-$ neighbors to superscript~$+$ neighbors, this means
$(c^-_{\ell},a^+_r) \in S_i$ for some neighbor 
$c^-_{\ell}$ that $a_r$ prefers to $b^-_{\ell}$, i.e., $a$ prefers $c$ to $b$. Thus $a \in A'_+$.

We have $A'_+ = Z_A \cup (A'_+ \setminus Z_A)$ and $A'_+ \setminus Z_A \subseteq A_+$ (by the while loop termination condition).
Lemma~\ref{new-correct2} and Lemma~\ref{correct3} showed that for $a \in Z_A \cup (A_+\cap A'_+)$, 
$a$ likes $M(a)$ at least as much as $L(a)$ and we showed in the above paragraph that each  $a \in A'_+$ prefers $L(a)$ to $b$.
Thus $\wt_{M}(a,b) = 0$. Since we set $\alpha_a = 0$ for $a \in Z_A$ and $\alpha_a = 1$ for $a \in A_+$,
we have $\alpha_a + \alpha_b \ge 0 = \wt_{M}(a,b)$.

We now need to show $\alpha_a + \alpha_b \ge \wt_{M}(a,b)$ holds for all edges $(a,b)$ in $G \setminus U_B$.
Recall the witness $\vec{\gamma}$ defined in the proof of Theorem~\ref{thm:M-U-M-L} to show the popularity of $M$ in 
the subgraph $G \setminus U_B$. Observe that it is only for vertices $u$ in $Z$ that we have $\alpha_u \ne \gamma_u$.
Moreover, $\alpha_a > \gamma_a$ for $a \in Z_A$. 

For $b \in Z_B$, we have $\alpha_b = 0$ while $\gamma_b = 1$. Thus we have to worry about edges $(a,b)$ in $G \setminus U_B$
where $b \in Z_B$ and check that $\wt_{M}(a,b) \le \alpha_a + \alpha_b$. Edges in $G \setminus U_B$ that are {\em not} incident to $Z_B$ are covered by $\vec{\alpha}$ 
since $\vec{\gamma}$ covers these edges and $\alpha_u \ge \gamma_u$ for all $u \notin Z_B$. 

Let $b \in Z_B \subseteq B_+ \cap B'_{-}$. We consider the following three possibilities for the vertex $a$.
\begin{enumerate}
\item Suppose $a \in U_A \cup Z_A$. For any $(a,b) \in (U_A \cup Z_A) \times B_+$, we have $\wt_{M}(a,b) \le \gamma_a + \gamma_b \le 0 + 1$. Because
$\wt_{M}(a,b)$ is an even number, this means $\wt_{M}(a,b) \le 0$. Since
$\alpha_a = 0$ for $a \in U_A\cup Z_A$ and $\alpha_b = 0$ for $b \in Z_B$, we have $\wt_{M}(a,b) \le 0 = \alpha_a + \alpha_b$.

\item Suppose $a \in A_{-} \setminus Z_A$. Then $a \in A'_{-}$ by the termination condition of the while-loop in our algorithm.
Since $\wt_{L}(x,y) \le \beta_x + \beta_y = -2$ for every edge $(x,y) \in A'_{-} \times B'_{-}$,  it follows that 
$b \in Z_B \subseteq B'_{-}$ prefers $L(b)$ to $a$ and similarly, $a \in A'_{-}$ prefers $L(a)$ to $b$.

We know from  Lemma~\ref{new-correct2} that $M(b) = L(b)$, so $b$ prefers $M(b)$ to $a$.
We know from Lemma~\ref{correct1} that 
$a$ likes $M(a)$ at least as much as $L(a)$, so $a$ prefers $M(a)$ to $b$. Thus $\wt_{M}(a,b) = -2 < \alpha_a + \alpha_b$
since $\alpha_a = -1$ and $\alpha_b = 0$.

\item Suppose $a \in A_+$. There are two subcases here: (i)~$a \in A'_{-}$ and (ii)~$a \in A'_+$. In subcase~(i), 
$\wt_{L}(a,b) \le \beta_a + \beta_b = -2$. Since $M(b) = L(b)$ (by  Lemma~\ref{new-correct2}), it means that $b$ prefers 
$M(b)$ to $a$. Hence $\wt_{M}(a,b) \le 0 < \alpha_a + \alpha_b$ since $\alpha_a = 1$ and $\alpha_b = 0$ here.

Consider subcase~(ii). We have $\wt_{L}(a,b) \le \beta_a + \beta_b = 0$. So either
(1)~$b$ prefers $L(b)$ to $a$ or (2)~$a$ prefers $L(a)$ to $b$. In case~(1), we have $\wt_{M}(a,b) \le 0$ since $M(b) = L(b)$
(by Lemma~\ref{new-correct2}). In case~(2) also, we have $\wt_{M}(a,b) \le 0$ since $a$ likes $M(a)$ at least as much as $L(a)$ (by Lemma~\ref{correct3}).
So in both cases we have $\wt_{M}(a,b) \le 0 < \alpha_a + \alpha_b$ since $\alpha_a = 1$ and $\alpha_b = 0$ here.
\end{enumerate}
Thus $\vec{\alpha}$ is a witness of $M$'s popularity (by Theorem~\ref{thm:witness}). Hence $M$ is popular in $G$. \qed
\end{proof}

Since $M$ is $A$-popular (recall that it uses only valid edges), Theorem~\ref{thm:final} immediately implies that $M$ is fully popular in $G$.
Moreover, $M$ is a {\em max-size} fully popular matching in $G$, as shown below.

\begin{lemma}
\label{lemma:last}
The matching $M$ is a max-size fully popular matching in $G = (A\cup B, E)$.
\end{lemma}
\begin{proof}
Observe that $U_A$ is the set of agents left unmatched in the matching $M$. We claim that all the agents in $U_A$ are left 
unmatched in any fully popular matching $N$ in $G$. We will use the fact that the matching $S_i$ is the 
$(A_L\cup B_L)$-optimal matching in the lattice ${\cal D}_i$ to prove this claim. 

Let $a\in A$ be such that $(a_{\ell}^-,a_r^+) \in S_i$. Since $S_i$ is the 
$(A_L\cup B_L)$-optimal matching in the lattice~${\cal D}_i$, if $a_{\ell}$ is matched to its least preferred neighbor $a^+_r$ in $S_i$, then $a_{\ell}$ 
cannot be matched to a better neighbor in the realization $N^*_{\alpha}$ of $N$, for any witness $\vec{\alpha}$ of $N$. In other words,
$(a_{\ell}^-,a_r^+) \in N^*_{\alpha}$. Thus $a$ is left unmatched in $N$ as well.
Hence $|M| = |A\setminus U_A| \ge |N|$. 
\qed
\end{proof}

\paragraph{\bf Running time of the algorithm.} The set of popular edges can be computed in linear time~\cite{CK16} and similarly,
the set of valid edges can be computed in linear time~\cite{AIKM07}.
The Gale-Shapley algorithm with forbidden edges in $H$ can be implemented to run in time linear in the size of $H$~\cite{GI89}, which is $O(m+n)$, 
where $|E| = m$ and $|A\cup B| = n$. 

Let us consider the time  taken by Algorithm~\ref{alg:fully-popular} in line~6 added up over all iterations. This is the same as running  
the Gale-Shapley algorithm with forbidden edges. 
More explicitly, when $S_{i-1}$ is modified to $S_i$ in the $i$-th iteration, all the intermediate edges considered while modifying $S_{i-1}$ to $S_i$ 
(these edges are now forbidden) are henceforth deleted from the graph. Thus the total time taken by Algorithm~\ref{alg:fully-popular} in line~6 added 
up over all the iterations is linear in the size of $H$.

It is easy to see that updating the sets $A_-, A'_+,B'_-, B_+$  takes $O(1)$ time per vertex since any vertex can move at most once from $A_+$ to $A_-$ (similarly, from $A'_-$ to $A'_+$ and from $B'_+$ to $B'_-$ and from $B_-$ to $B_+$). 
We need to check at the start of each iteration if there is an unmarked vertex in $(A_- \cap A'_+) \cup (B_+ \cap B'_-)$.
This can be implemented efficiently by maintaining a list of vertices $u$ such that both $(u^-_{\ell},\ast)$ and $(\ast,u^+_r)$ are in our
matching. 

So for each vertex $u$, whenever (i)~$u_{\ell}$ starts proposing to superscript $+$ neighbors or (ii)~$u_r$ receives a 
proposal from a superscript $-$ neighbor, we check if $u_{\ell}$ and $u_r$ are in opposite states. If so, then 
$u$ is added to the end of this list. We use a pointer that traverses this list once from left to right during the entire course of 
the algorithm. At the start of each iteration, we start from the current position of this pointer and traverse rightwards in the list 
searching for a vertex that is still unmarked.
Thus we can efficiently check if  $(A_- \cap A'_+) \cup (B_+ \cap B'_-)$ has an unmarked vertex or not. Hence our algorithm can be 
implemented to run in linear time. Thus Theorem~\ref{thm:algo} follows.

\MainTheorem*

\paragraph{\bf Acknowledgements.} Supported by the DAE, Government of India, under project no. RTI4001.
Thanks to Yuri Faenza for discussions that led to this problem and his helpful comments on the manuscript. 
Thanks to the reviewers of the conference version of this paper for their suggestions on improving the presentation.

\end{document}